\documentclass[submission,copyright,creativecommons]{eptcs}
\providecommand{\event}{SOS 2007} 

\usepackage{iftex}

\ifpdf
  \usepackage{underscore}         
  \usepackage[T1]{fontenc}        
\else
  \usepackage{breakurl}           
\fi

\usepackage{times}
\usepackage{soul}
\usepackage{url}
\usepackage{hyperref}
\usepackage[utf8]{inputenc}
\usepackage[small]{caption}
\usepackage{graphicx}
\usepackage{amsmath}
\usepackage{amsthm}
\usepackage{booktabs}
\usepackage{algorithm}
\usepackage{algorithmic}
\usepackage[switch]{lineno}

\usepackage{symbols}
\usepackage[capitalise]{cleveref}
\usepackage{tikz}
\usepackage[all]{foreign}

\newtheorem{definition}{Definition}
\newtheorem{lemma}{Lemma}
\newtheorem{theorem}{Theorem}

\title{Synthesis of Timeline-Based Planning \\ Strategies Avoiding
Determinization\thanks{This work is partially supported by the INdAM-GNCS
Project \emph{Analisi simbolica e numerica di sistemi ciberfisici} (project
n.~CUP\_E53C22001930001).}}

\author{Renato Acampora
\institute{University of Udine, Italy}
\email{renato.acampora@uniud.it}
\and
Dario Della Monica
\institute{University of Udine, Italy}
\email{dario.dellamonica@uniud.it}
\and
Luca Geatti
\institute{University of Udine, Italy}
\email{luca.geatti@uniud.it}
\and
Nicola Gigante
\institute{Free University of Bozen-Bolzano, Italy}
\email{nicola.gigante@unibz.it}
\and Angelo Montanari
\institute{University of Udine, Italy}
\email{angelo.montanari@uniud.it}
\and Pietro Sala
\institute{University of Verona, Italy}
\email{pietro.sala@univr.it}
}

\def\titlerunning{Synthesis of Timeline-Based Planning Strategies Avoiding Determinization}
\def\authorrunning{Acampora et al.}
\begin{document}
\maketitle

\begin{abstract}
    Qualitative timeline-based planning models domains as sets of independent, but interacting, components whose behaviors over time, the timelines, are governed by sets of qualitative temporal constraints (ordering relations), called synchronization rules.
Its plan-existence problem has been shown to be \PSPACE-complete; in particular,
\PSPACE-membership has been proved via reduction to the nonemptiness problem
for nondeterministic finite automata.
However, nondeterministic automata cannot be directly used to synthesize planning strategies as a costly determinization step is needed.
%
%
In this
paper, we identify a large fragment of qualitative timeline-based planning whose
plan-existence problem can be directly mapped into the nonemptiness problem of
deterministic finite automata, which can then 
be exploited to synthesize strategies.
%
%
In addition, we identify a  maximal subset of Allen's relations that fits into such a deterministic fragment. 
\end{abstract}


\section{Introduction}
\label{sec:intro}

\emph{Timeline-based} planning is an approach that originally emerged and developed in
the context of planning and scheduling of \emph{space}
operations~\cite{Muscettola94}. In contrast to common action-based formalisms,
such as PDDL~\cite{FoxL03}, timeline-based languages do not make a distinction between
actions, states, and goals. Rather, the domain is modeled as a set of
independent, but interacting, components whose behavior over time, the
timelines, is governed by a set of temporal constraints. 
It is worth pointing out that timeline-based planning was born with an application-oriented flavor, with various successful stories, and only relatively recently some foundational work about its expressiveness and complexity has been produced.
The present paper aims at bringing back theory to practice by identifying expressive enough and computationally well-behaved fragments. 

\sloppy
Timeline-based planning has been successfully employed by planning systems developed at
NASA~\cite{ChienRKSEMESFBST00,ChienSTCRCDLMFTHDSUBBGGDBDI04} and
at ESA~\cite{FratiniCORD11}  for both short- to long-term mission planning and
on-board autonomy. More recently, timeline-based planning systems such as
PLATINUm~\cite{UmbricoCMO17} are being employed in collaborative robotics
applications~\cite{UmbricoCO23}. All these applications share a deep reliance 
on \emph{temporal reasoning} and the need for a tight integration of planning
with \emph{execution}, both features of the timeline-based framework. 
The latter feature is usually achieved by the use of \emph{flexible timelines},
which represent a set of possible executions of the system that differ in the
precise timing of the events, hence handling the intrinsic \emph{temporal
uncertainty} of the environment. A formal account of timeline-based planning
with uncertainty has been provided by \cite{CialdeaMayerOU16}, and much
theoretical research followed, including
\emph{complexity}~\cite{BozzelliMMP18b,BozzelliMMP18a,GiganteMCO17} and
\emph{expressiveness}~\cite{DellaMonicaGMS18,GiganteMMO16} analyses, based on
such a formalization, which is the one we use here as well.

To extend the reactivity and adaptability of timeline-based systems
beyond temporal uncertainty, the framework of \emph{timeline-based games} has
been recently proposed. In timeline-based games, the system player tries to build a
set of timelines satisfying the constraints independently from the choices of
the environment player. This framework allows one to handle general
nondeterministic environments in the timeline-based setting. However, this
expressive power comes at the cost of increasing the complexity of the problem.
While the plan-existence problem for timeline-based planning is
EXPTIME-complete~\cite{GiganteMCO17}, deciding the existence of strategies for
timeline-based games is 2EXPTIME-complete~\cite{GiganteMOCR20}, and a controller
synthesis algorithm exists that runs in doubly exponential
time~\cite{AcamporaGGMP2022}. 

Such a high complexity motivates the search for simpler fragments 
that can nevertheless be useful in practical scenarios. One of these is the
\emph{qualitative} fragment, where temporal constraints only concern the
relative order between events and not their 
distance. The qualitative fragment already proved itself to be easier for the 
plan-existence problem, being PSPACE-complete~\cite{DellaMonicaGTM20}, and 
this makes it a natural candidate for the search of a good fragment for the 
strategy existence problem. 

A \emph{deterministic} arena is crucial to synthesize a
non-clairvoyant strategy in \emph{reactive synthesis} problems~(see, for instance,
\cite{PnueliR89}). However, determinizing the nondeterministic (exponentially
sized) automaton built
for the qualitative case in \cite{DellaMonicaGTM20} would cause an exponential
blowup, thus resulting in a procedure of doubly-exponential complexity.
In this paper, we show that, by imposing some natural restrictions on the set of temporal constraints of the qualitative fragment, it is possible to lower the complexity of the strategy existence problem to EXPTIME.
%
We show that, on the one hand, these restrictions are 
\emph{sufficient} to directly synthesize a \emph{deterministic} finite automaton (DFA) of singly-exponential size, thus usable as an arena to play the game in an asymptotically optimal way, and, on the other hand, the resulting fragment is expressive enough to capture a large subset of Allen's relations~\cite{Allen83}, defined in \cref{sec:allen}. 

The rest of the paper is organized as follows. \cref{sec:preliminaries} recalls
some background knowledge on timeline-based planning. \cref{sec:fragment} defines
the considered fragment, that directly maps into a DFA  
of singly exponential size. \cref{sec:wordsToPlans} gives a word encoding of 
timelines, and vice versa. \cref{sec:dfaForPlans} builds an automaton 
to recognize plans, and \cref{sec:dfaForSolutions} shows how to construct 
an automaton that accepts solution plans. \cref{sec:allen} 
identifies the maximal subset of Allen's relations which is captured by
the fragment of \cref{sec:fragment}. Finally, \cref{sec:conclusion} summarizes
the main contributions of the work and discusses possible future developments.



\section{Background}
\label{sec:preliminaries}

In this section, we recall the basic notions of timeline-based
planning and of its qualitative variant.

\subsection{Timeline-Based Planning}

The key element of the framework is the notion of \emph{state variable}. 
Let $\N^+$ be the set of positive natural numbers.

\begin{definition}[State variable]
  \label{def:statevar}
  A \emph{state variable} is a tuple $x=(V_x,T_x,D_x)$, where:
  \begin{itemize}
  \item $V_x$ is the \emph{finite domain} of the variable;
  \item $T_x:V_x\to2^{V_x}$ is the \emph{value transition function}, which maps
        each value $v\in V_x$ to the set of values that can (immediately) follow it;
  \item $D_x:V_x\to\N^+\times(\N^+\cup\set{+\infty})$ is a function that maps
        each $v\in V_x$ to the pair $(d^{x=v}_{min},d^{x=v}_{max})$ of
        minimum and maximum durations allowed for intervals where $x=v$.
  \end{itemize}
\end{definition}

A \emph{timeline} is a finite sequence of \emph{tokens}, each denoting a value $v$ and (the duration of) a time interval $d$, that describes how a state variable $x$ behaves over time. 

\begin{definition}[Tokens and timelines]
  \label{def:timeline}
  A \emph{token} for $x$ is a tuple $\tau=(x,v,d)$, where $x$ is a state
  variable, $v\in V_x$ is the value held by the variable, and $d\in\N^+$ is
  the \emph{duration} of the token, with $D_x(v) = (d^{x=v}_{min},d^{x=v}_{max})$ 
  and $d^{x=v}_{min} \le d \le d^{x=v}_{max}$.
  A \emph{timeline} for a state variable $x$ is a finite sequence
  $\timeline=\seq{\tau_1,\ldots,\tau_k}$ of tokens for $x$, for some
  $k\in\N$, such that, for any $1\le i < k$, if $\tau_i = (x,v_i,d_i)$,
  then $v_{i+1} \in T_x(v_i)$.
\end{definition}

For any timeline $\timeline = \seq{\tau_1,\ldots,\tau_k}$ and any
token $\tau_i=(x,v_i,d_i)$ in $\timeline$, we define the functions
$\starttime(\timeline,i) = \sum_{j=1}^{i-1} d_j$ and
$\endtime(\timeline,i) = \starttime(\timeline,i) + d_i$. 
We call the \emph{horizon} of $\timeline$ the end time of the last token in
$\timeline$, that is, $\endtime(\timeline, k)$.  We write
$\starttime(\tau_i)$ and $\endtime(\tau_i)$ to indicate
$\starttime(\timeline,i)$ and $\endtime(\timeline,i)$, respectively, when
there is no ambiguity.

The overall behavior of state variables is subject to
a set of temporal constraints known as \emph{synchronization rules} (or
simply \emph{rules}). We start by defining their basic building blocks.
Let $\toknames$ be a finite set of \emph{token names}. \emph{Atoms} are
formulas of the following form:
\begin{equation*}
\begin{split}
  \mathit{atom} &\coloneqq \mathit{term} \before_{l,u} \mathit{term} \mid \mathit{term} \before*_{l,u} \mathit{term} \\
  \mathit{term} &\coloneqq \tokstart(a) \| \tokend(a) \| t
\end{split}
\end{equation*}
where $a\in\toknames$, $l,t\in\N$, and $u\in\N\cup\set{+\infty}$. 
As an example, atom $\tokstart(a)\before_{l,u}\tokend(b)$ (resp.,
$\tokstart(a)\before*_{l,u}\tokend(b)$) relates tokens $a$ and $b$ by
stating that the end of $b$ cannot precede (resp., must succeed) the
beginning of $a$, and the distance between these two endpoints must
be at least $l$ and at most $u$.  An atom
$\mathit{term}\before_{l,u}\mathit{term}$, with $l=0$ and $u=+\infty$, is
\emph{qualitative} (the subscript is usually omitted in this case).


An \emph{existential statement} $\E$ is a constraint of the form:
  \[ \exists a_1[x_1=v_1]a_2[x_2=v_2] \ldots a_n[x_n=v_n]. \ \rulebody \]
where $x_1, \ldots, x_n$ are state variables, $v_1, \ldots, v_n$ are values,
with $v_i\in V_{x_i}$, $a_1,\ldots,a_n$ are symbols from the set \toknames of
token names, and $\rulebody$ is a finite conjunction of \emph{atoms}, involving only 
tokens $a_1,\ldots,a_n$, plus, possibly, the
  \emph{trigger token} (usually denoted by $a_0$) of the \emph{synchronization
    rule} in which the existential statement is embedded, as described below.\footnote{%
  W.l.o.g., we assume that if a token $a$ appears in the
  quantification prefix $\exists a_1[x_1=v_1]a_2[x_2=v_2] \ldots a_n[x_n=v_n]$ of \E, then at least one among $\tokstart(a)$ and $\tokend(a)$ occurs in one of its atoms. 
}

Intuitively, an existential statement asks for the existence of tokens
$a_1,a_2,\dots,a_n$ whose state variables take the corresponding values
$v_1,v_2,\dots,v_n$ and are such that their start and end times satisfy the
atoms in $\rulebody$.

\emph{Synchronization rules} are clauses of one of
the following forms:
\begin{align*}
    a_0[x_0=v_0]\implies \E_1\lor\E_2\lor\ldots\lor\E_k \\ 
    \top \implies \E_1\lor\E_2\lor\ldots\lor\E_k
\end{align*}
where $a_0 \in \toknames$, $x_0$ is a state variable, $v_0 \in V_{x_0}$,
and $\E_i$ is an existential statement, for each $1\le i\le k$. 
In the former case, $a_0[x_0=v_0]$ is called \emph{trigger} and $a_0$ is the \emph{trigger token}, and the rule
is considered \emph{satisfied} if \emph{for all} the tokens $a_0$ for which
the variable $x_0$ takes the value $v_0$, at least one of the existential
statements is satisfied.
In the latter case, the rule is said to be \emph{triggerless}, and it
states the truth of the body without any precondition.%
\footnote{%
  W.l.o.g., for non-triggerless rules, we assume that both $\tokstart(a_0)$
  and $\tokend(a_0)$ occur in all of its existential statements.
}
We refer the reader to~\cite{CialdeaMayerOU16} for a formal account of the semantics
of the rules.

A \emph{timeline-based planning problem} consists of a set of state variables
and a set of rules that represent the problem domain and the goal.

\begin{definition}[Timeline-based planning problem]
  A \emph{timeline-based planning problem} is defined as a pair $P = (\SV, S)$,
  where $\SV$ is a set of state variables and $S$ is a set of synchronization
  rules involving state variables in $\SV$.
\end{definition}

A \emph{solution plan} for a given timeline-based planning problem is a set of
timelines, one for each state variable, that satisfies all the synchronization
rules.

\begin{definition}[Plan and solution plan]
  \label{def:plan}
  A \emph{plan} over a set of state variables $\SV$ is a finite set of timelines
  with the same horizon, one for each state variable $x\in\SV$.
  A \emph{solution plan} for a timeline-based planning problem $P=(\SV,S)$ is a
  plan over \SV such that all the rules in $S$ are satisfied. 
\end{definition}

The problem of determining whether a solution plan exists for a given
timeline-based planning problem is \EXPSPACE-complete \cite{GiganteMCO17}.
%

\begin{definition}[Qualitative timeline-based planning]
\label{def:qual:tbp}
    A timeline-based planning problem $P=(\SV,S)$ is said to be \emph{qualitative} if the following conditions hold: 
    \begin{enumerate}
    \item $D_x(v)=(1,+\infty)$, for all state variables $x\in\SV$ and $v\in
      V_x$.
    \item all synchronization rules in $S$ involve only \emph{qualitative}
      atoms.
    \end{enumerate}
\end{definition}


%
%

Unlike timeline-based planning, such a qualitative variant is
PSPACE-complete~\cite{DellaMonicaGTM20}.  A reduction of qualitative
timeline-based planning to the nonemptiness problem for non-deterministic
finite automata (NFA) has been provided in~\cite{DellaMonicaGTM20}.

\section{A Well-Behaved Fragment}
\label{sec:fragment}


In this section, we introduce a meaningful fragment of qualitative timeline-based
planning for which we will show that it is
possible to construct DFAs of singly exponential size. 

The fragment is characterized  by means of some conditions on the admissible patterns of synchronization rules (eager rules). The distinctive feature of eager rules is that they can be checked using an eager/greedy strategy, that is, when a relevant event (start/end of a token involved in some atom) occurs, we are guaranteed that the starting/ending point of such a token is useful for rule satisfaction. Instead, in case of non-eager rules, it may happen that a relevant event happens that is not useful for rule satisfaction: some analogous event in the future will be.

W.l.o.g., we assume that no constraint of the forms $\tokstart(a) \before \tokend(a)$ and $\tokstart(a) \before* \tokend(a)$ occurs explicitly in synchronization rules, even though they hold tacitly, as they follow from the definition of token (Definition \ref{def:timeline}).

As a preliminary step, we define a sort of transitive closure of a clause.
First, by slightly abusing the notation, we identify a clause \rulebody 
with the finite set of atoms occurring in it. Let $t$, $t_1$, $t_2$, $t_3$ be
terms of the form $\tokstart(a)$ or $\tokend(a)$, with $a \in \toknames$.
We denote by $\hat \rulebody$ the \emph{transitive closure} of $\rulebody$,
defined as the smallest set of atoms including $\rulebody$ and such that:
\begin{enumerate*}[label={\it (\roman*)}]
\item if term $t$ occurs in $\rulebody$, then atom $t \before t$ belongs to
  $\hat \rulebody$,
\item if terms $\tokstart(a)$ and $\tokend(a)$ both occur in $\rulebody$ for
  some token name $a$, then atom $\tokstart(a) \before* \tokend(a)$ belongs to
  $\hat \rulebody$,
\item if atom $t_1 \before* t_2$ belongs to $\hat
  \rulebody$, then atom $t_1 \before t_2$ belongs to
  $\hat \rulebody$ as well,
\item if atoms $t_1 \before t_2$ and $t_2
  \before t_3$ belong to $\hat \rulebody$, then atom $t_1\before t_3$
  belongs to $\hat \rulebody$ as well,
\item if atoms $t_1 \before* t_2$ and $t_2
  \before t_3$ belong to $\hat \rulebody$, then atom $t_1 \before* t_3$ 
  belongs to $\hat \rulebody$ as well,
\item if atoms $t_1 \before t_2$  and $t_2
  \before* t_3$ belong to $\hat \rulebody$, then atom $t_1 \before* t_3$ belongs to $\hat \rulebody$ as well.
\end{enumerate*}%
\footnote{%
  W.l.o.g., we assume that $\hat\rulebody$ is consistent, \ie it admits at
  least a solution. We point out that this check can be done in polynomial
  time, since it is an instance of linear programming.
}

Notice that, in some particular cases, condition {\it (ii)} may introduce in the closure of a clause atoms of the form $\tokstart(a) \before* \tokend(a)$, which, according to our assumption, do not belong to any clause.

Let us  now define the core notion of \emph{eager rule}.

\begin{definition}[Eager rules]
\label{def:eager:rule}
  Let $\Rule$ be a synchronization rule and let $\rulebody_1, \ldots,
  \rulebody_k$ be the clauses occurring in its existential statements.
  We say that $\Rule$ is \emph{eager} if and only if, for all $\rulebody \in \{
  \rulebody_1, \ldots, \rulebody_k \}$ and   
  $a_1,a_2 \in \toknames$ appearing in $\rulebody$, 
  the following conditions hold:
  \begin{enumerate}
    \item  if both $a_1$ and $a_2$ are non-trigger tokens and $\{ \tokstart(a_2) \before \tokend(a_1), \tokend(a_1)
      \before \tokend(a_2) \} \subseteq \hat \rulebody$, then
      $\tokend(a_1) \before \tokstart(a_2) \in \hat \rulebody$
      (i.e., the end of $a_1$ and the start of $a_2$ coincide),
      \item if $a_1$ is either a
      trigger token or a non-trigger one, 
      $a_2$ is a non-trigger token, and
      $\{ \tokstart(a_2) \before \tokstart(a_1), \tokstart(a_1)
      \before \tokend(a_2) \} \subseteq \hat \rulebody$, then
      $\tokstart(a_1) \before \tokstart(a_2) \in \hat \rulebody$ (i.e., $a_1$ and $a_2$ start together), and
    \item if $a_1$ is a trigger token, $a_2$ is a non-trigger one,     
     and $\{ \tokstart(a_1) \before \tokstart(a_2), \allowbreak \tokend(a_1) \before \tokend(a_2) \} \subseteq \hat \rulebody$, then $\tokstart(a_2) \before \tokstart(a_1)  \in \hat \rulebody$
     (i.e., $a_1$ and $a_2$ start together).
  \end{enumerate}
\end{definition}
%
%

We define the \emph{eager fragment} of a qualitative timeline-based
planning problem as the set of qualitative timeline-based planning problems
$P = (\SV, S)$ such that $S$ contains only eager rules.

An explanation of the restrictions in~\cref{def:eager:rule} is due. Given
a non-trigger token $a_2$, Condition 1 forces any other non-trigger token $a_1$ ending during $a_2$ (that is, such that $\tokstart(a_2) \le \tokend(a_1) \le
\tokend(a_2)$) to end exactly when $a_2$ starts, while Condition 2 forces any other (trigger or non-trigger) token $a_1$ starting during $a_2$ (that is, such that $\tokstart(a_2) \le \tokstart(a_1) \le \tokend(a_2)$) to start simultaneously to $a_2$. Finally, whenever a non-trigger token $a_2$ starts during a trigger token $a_1$ and ends not before the end of $a_1$, Condition 3 forces the two tokens to start at the same time.

Conditions 1, 2, and 3 suffice to obtain a singly exponential DFA, whose construction will be illustrated in the next
sections. We give here a short intuitive account of the rationale of the above conditions.

Consider the following rule:
\begin{align*}
  a_0[x_0=v_0]&\implies \exists a_1[x_1=v_1]. \\
  &(\tokstart(a_0) = \tokstart(a_1) \wedge \tokend(a_0) \before \tokend(a_1)),
\end{align*}
where $\tokstart(a_0) = \tokstart(a_1)$ is an abbreviation for $\tokstart(a_0)
\before \tokstart(a_1) \wedge \tokstart(a_1) \before \tokstart(a_0)$.
This rule is eager because Conditions 1, 2, and 3 are fulfilled; in
particular, we have that $\tokstart(a_0) = \tokstart(a_1)$. This is crucial
for any DFA $\autom$ recognizing solution plans, because, when $\autom$
reads the event $\tokstart(a_0)$, it can \emph{eagerly} and
\emph{deterministically} go to a state representing the fact that both
$\tokstart(a_0)$ and $\tokstart(a_1)$ have happened.  Moreover, if later it
reads the event $\tokend(a_1)$, but it has not read $\tokend(a_0)$ yet, then
it transitions to a rejecting state, that is, a state from which it cannot
accept any plan.

Let us provide now an example of a non-eager rule that cannot be checked in an
eager/greedy fashion.
Consider the rule obtained from the above one by replacing $=$ with $\before$:
\begin{align*}
  a_0[x_0=v_0]&\implies \exists a_1[x_1=v_1]. \\
  &(\tokstart(a_0) \before \tokstart(a_1) \wedge \tokend(a_0) \before \tokend(a_1)).
\end{align*}
This rule is \emph{not} eager, because atom $\tokstart(a_1)
\before \tokstart(a_0)$ does not belong to $\hat\rulebody$ (Condition 3 is
violated). Indeed, for this rule, a DFA $\autom$ that first reads event
$\tokstart(a_0)$, but not  $\tokstart(a_1)$, and
then, strictly after, reads event $\tokstart(a_1)$ has to
\emph{nondeterministically} guess the order between the end of such a token
$a_1$ and the end of $a_0$, making  the construction of an automaton of
singly exponential size impossible in the general case. Indeed, if token $a_1$ ends before token $a_0$, the rule is not satisfied, but we cannot exclude the existence of another token for $x_1=v_1$ that starts after that one and ends after the end of $a_0$, thus satisfying the rule. 


We conclude by showing that excluding 
constraints of the forms $\tokstart(a) \before \tokend(a)$ and $\tokstart(a) \before* \tokend(a)$ from clauses makes it sometimes possible to turn an otherwise non-eager rule into an eager one.
As an example, rule $a_0[x_0=v_0]\implies \exists a_1[x_1=v_1].
(\tokstart(a_1) \before* \tokend(a_1) \wedge \tokstart(a_0) = \tokend(a_1))$
is not eager (Condition 2 is violated);
however, it can be rewritten as $a_0[x_0=v_0]\implies \exists
a_1[x_1=v_1]. \tokstart(a_0) = \tokend(a_1)$, which is eager.

In what follows, we give a reduction from the plan-existence problem for
the eager fragment of the qualitative timeline-based planning problem to
the nonemptiness problem of DFAs of \emph{singly exponential} size with 
respect to the original problem.  The approach is inspired by those in
\cite{DellaMonicaGTM20,DellaMonicaGMS18} for non-eager timeline-based
planning problems, where an NFA of exponential size is built for any
timeline-based planning problem.  However, the reductions presented there
use nondeterministic automata, which cannot be used as arenas to solve
timeline-based games without a previous determinization step that would
cause a second exponential blowup.

First, we show how to encode timelines and plans as finite words, and vice versa
(Section~\ref{sec:wordsToPlans}).  Then, given a planning problem $P$, we
show how to build a DFA whose language encodes the set of solution plans
for $P$.  The DFA consists of the intersection of two DFAs: one aims at
verifying the constraint on the alternation of token values expressed by
functions $T_x$, for $x \in \SV$, as well as that the word correctly
encodes a plan over \SV (Section~\ref{sec:dfaForPlans}); the other one
verifies that the encoded plan is indeed a solution plan for $P$
(Section~\ref{sec:dfaForSolutions}).

From now on, we consider only qualitative timeline-based planning problems
belonging to the eager fragment and, for the sake of brevity, we sometimes
refer to them simply as \emph{planning problems}.




\section{From Plans to Finite Words and \textit{Vice Versa}}
\label{sec:wordsToPlans}

In this section, as a first step in the construction of the DFA
corresponding to an eager qualitative timeline-based planning problem, we
show how to encode timelines and plans as \emph{words} that can be
recognized by an automaton, and \viceversa.

Let $P = (\SV, S)$ be an eager qualitative timeline-based planning
problem, and let $V = \cup_{x\in \SV}V_x$.
We define the \emph{initial alphabet} $\Sigma_\SV^I$ as $(\{ - \} \times
V)^\SV$, that is the set of functions from $\SV$ to $(\{ - \} \times V)$.%
\footnote{%
  The symbol $\{ - \}$ is a technicality that allows us to consider pairs
  instead of just values in $V$.
}
Similarly, we define the \emph{non-initial alphabet} $\Sigma_\SV^N$ as $((V
\times V) \cup \{ \circlearrowleft \} )^\SV$, where the pairs $(v,v') \in
V \times V$ are supposed to represent the value $v$ of the token that just
ended and the value $v'$ of the token that has just started, and
$\circlearrowleft$ represents the fact that the value for the state
variable has not changed.
The \emph{input alphabet} (or, simply, \emph{alphabet}) associated with
$\SV$ and denoted by $\Sigma_\SV$ is the union $\Sigma_\SV^I \cup
\Sigma_\SV^N$.
%
%
%
%
Observe that the size of the alphabet $\Sigma_\SV$ is at most exponential
in the size of $\SV$, precisely $\abs{\Sigma_\SV} = \abs{\Sigma_\SV^I}
+ \abs{\Sigma_\SV^N} = \abs{V}^{\abs{\SV}} + (\abs{V}^2 + 1)^{\abs{\SV}}$.
%



We now show how to encode the basic structure%
\footnote{%
  With ``basic structure'' we refer to the fact that, in this section, we
  neither take into account the transition functions $T_x$ of state
  variables nor their domains $V_x$ (\cf~\cref{def:statevar}), which will be
  dealt with in~\cref{sec:dfaForPlans}.
}
underlying each plan over \SV as a word in $\Sigma_\SV^I \cdot
(\Sigma_\SV^{N})^* \cup \{ \varepsilon \}$, where $\varepsilon$ is the
empty word (and corresponds to the empty plan), $(\Sigma_\SV^{N})^*$ is the
Kleene's closure of $\Sigma_\SV^{N}$, and $\cdot$ denotes the concatenation
symbol.
Intuitively, let $\nu$ be the symbol at position $i$ of a word $\sigma \in
\Sigma_\SV^I \cdot (\Sigma_\SV^{N})^* \cup \{ \varepsilon \}$.
Then, if $\nu(x) = (v, v')$ for some $x \in \SV$, then at time $i$ a new token
begins in the timeline for $x$ with value $v'$; instead, if $\nu(x) =
\circlearrowleft$, then no change happens at time $i$ in the timeline for $x$,
meaning that no token ends at that time point in the timeline for $x$.
The value $v$ of the token ending at time $i$ will be used later in the
construction of the automata.

We remark that not all words in $\Sigma_\SV^I \cdot (\Sigma_\SV^{N})^* \cup
\{ \varepsilon \}$ correspond to plans over \SV: for a word to correctly
encode a plan, the information carried by the word about the value of
a starting token and the one associated to the end of the same token must
coincide.
Formally, given a word $\sigma=\seq{\sigma_0,\ldots,\sigma_{|\sigma|-1}}
\in \Sigma_\SV^I \cdot (\Sigma_\SV^{N})^* \cup \{ \varepsilon \}$ and
a state variable $x \in \SV$, let $\mathit{changes}(x)
= (i^x_0,i^x_1,\ldots,i^x_{k^x-1})$, for some $k^x \in \N$, be the
increasing sequence of positions where $x$ changes, \ie $\sigma_i(x) \neq
\circlearrowleft$ if and only if $i \in \mathit{changes}(x)$, for all $i
\in \{ 0, \ldots, |\sigma|-1 \}$.  We denote by $v^x_{i}$ and $\hat
v^x_{i}$ the first and the second component of $\sigma_i(x)$, respectively,
for all $x \in \SV$ and $i \in \mathit{changes}(x)$.
We omit superscripts $^x$ when there is no risk of ambiguity.

\begin{definition}[Words weakly-encoding plans]
\label{def:weak:encode}
  Let $\sigma \in \Sigma_\SV^I \cdot (\Sigma_\SV^{N})^*$ and let
  $\mathit{changes}(x) = (i_0,i_1,\ldots,i_{k-1})$.  We say that
  \emph{$\sigma$ weakly-encodes a plan over \SV} if $\hat v_{i_{h-1}}
  = v_{i_{h}}$ for all $x \in \SV$ and $h \in \{ 1, \ldots, k-1 \}$.  If
  this is the case, then the \emph{plan induced by $\sigma$} is the set $\{
  \timeline_x \mid x \in \SV \}$, where $\timeline_x=\seq{(x,\hat
  v_{i_0},i_1-i_0), (x,\hat v_{i_1},i_2-i_1), \ldots, (x,\hat v_{i_{k-1}},
  i_k-i_{k-1})}$ and $i_k = |\sigma|$, for all $x \in \SV$.
\end{definition}

Intuitively, if a word weakly-encodes a plan, then it captures
the dynamics of a state variable modulo its domain and its transition function,
which will be taken care of in the next section.
%
%
%
A converse correspondence from plans to words can be
defined accordingly.

Before concluding the section, we introduce another notation that will come
handy later.
We denote by \eventsOnesigma the set of events (beginning/ending of a token)
occurring at a given time, encoded in the alphabet symbol $\sigma$.
Formally, \eventsOnesigma is the smallest set such that:
\begin{itemize}
\item if $\sigma(x) = (v,v') $ for some $x$, then
    $\{ \mathit{end}(x,v),\mathit{start}(x,v') \} \subseteq
    \eventsOnesigma$, and
\item if $\sigma(x) = (-,v')$ for some $x$, then $\mathit{start}(x,v') \in
  \eventsOnesigma$.
\end{itemize}



\section{DFA Accepting Plans}
\label{sec:dfaForPlans}

Given an eager qualitative timeline-based planning problem $P = (\SV, S)$, we
show how to build a DFA $\mathcal T_{\SV}$, of size at most exponential in the
size of 
$P$, accepting words that correctly encode plans over \SV, that is, words
that weakly-encode plans (\cf~\cref{def:weak:encode}) \emph{and} comply
with the constraints on the alternation of token values expressed by
functions $T_x$, for $x \in \SV$.
In the next section, we show how to obtain a DFA, of size at most
exponential in the size of $P$, that accepts exactly the \emph{solution
plans} for $P$.

For every planning problem $P = (\SV, S)$, the DFA $\mathcal T_{\SV}$ is
the tuple $\langle Q_{\SV}, \Sigma_\SV, \delta_{\SV}, q^0_{\SV}, F_{\SV}
\rangle$, whose components are defined as follows.
\begin{itemize}

\item $Q_{\SV}$ is the set of \emph{states} of $\mathcal{T}_{\SV}$.
  Intuitively, a state of $\mathcal{T}_{\SV}$ keeps track of the token values of
  the timelines at the current and the previous step of the run.
  Therefore, a state is a function mapping each state variable $x$ into a pair
  $(v,v')$, where $v'$ (resp., $v$) denotes the token value of timeline $x$ at
  the current (resp., previous) step.
  To formally define $Q_{\SV}$, we exploit the definition of alphabet
  $\Sigma_\SV$ from Section~\ref{sec:wordsToPlans}.
  Mostly, states are alphabet symbols, except for those functions $\sigma \in
  \Sigma_\SV$ assigning to at least one state variable $x \in \SV$
  value $\circlearrowleft$.
  For technical reasons, we also need a fresh \emph{initial state} $q^0_{\SV}$
  and a fresh \emph{rejecting sink state} \sinkT.

  Formally, $Q_{\SV} = \left(\Sigma_\SV \setminus \bar{Q}_{\SV} \right) \cup \{
  q^0_{\SV}, \sinkT \}$, where $\bar{Q}_{\SV} = \{ \sigma \in \Sigma_{\SV} \mid
  \sigma(x) = \circlearrowleft \text{ for some } x \in \SV \}$.
  Clearly, the size of $Q_{\SV}$ is at most as the size of $\Sigma_\SV$, which
  is in turn at most exponential in the size of $P$.

%
%

\item $\Sigma_\SV$ is the \emph{input alphabet}, defined as in
  Section~\ref{sec:wordsToPlans}.
\item $\delta_{\SV} : Q_{\SV} \times \Sigma_\SV \rightarrow Q_{\SV}$ is
  the \emph{transition function}.
  Towards a definition of $\delta_{\SV}$, we say that an alphabet symbol $\sigma
  \in \Sigma_\SV$ is \emph{compatible} with a state $\sigma_1 \in Q_{\SV}$ (we
  use for states the same symbols as for the alphabet, i.e., $\sigma, \sigma_1,
  \sigma_2, \ldots$, to stress the fact that states are closely related to
  alphabet symbols) if one of the following holds:
  \begin{enumerate*}[label={\it (\roman*)}]
  \item $\sigma_1 = q^0_{\SV}$ is the initial state and $\sigma \in
    \Sigma_\SV^I$ is an initial symbol such that for each $x \in \SV$ it holds
    that $\sigma(x) = (-,v)$ with $v \in V_x$;
  \item $\sigma_1 = (v,v') \in \Sigma_\SV \setminus \bar{Q}_{\SV}$ and
    $\sigma \in \Sigma_\SV^N$ is a non-initial symbol such that for each $x
    \in \SV$ either $\sigma(x) = \circlearrowleft$ or $\sigma(x) = (v', v'')$
    with $v'' \in T_x(v') \cap V_x$.
  \end{enumerate*}

%
%


  Now, $\delta_{\SV} : Q_{\SV} \times \Sigma_\SV \to Q_{\SV}$ is defined as
  follows.
  For all $\sigma_1 \in Q_{\SV}$ and $\sigma \in \Sigma_\SV$, if $\sigma$ is not
  compatible with $\sigma_1$ or $\sigma_1$ is the sink state (i.e., $\sigma_1 =
  \sinkT$), then $\delta(\sigma_1, \sigma) = \sinkT$; otherwise
  \begin{itemize}
  \item if $\sigma_1$ is the initial state (i.e., $\sigma_1 = q^0_{\SV}$), then
    $\delta(\sigma_1, \sigma) = \sigma$; in other words, in this case the
      automaton transitions to the state represented by the input letter;
  \item if $\sigma_1 \in \Sigma_\SV \setminus \bar{Q}_{\SV}$, then
    $\delta(\sigma_1, \sigma) = \sigma_2$, where $\sigma_2(x)
      = \sigma_1(x)$ if $\sigma(x) = \circlearrowleft$, and $\sigma_2(x)
      = \sigma(x)$ otherwise, for all $x \in \SV$; intuitively, the
      automaton transitions into a state keeping track of the updated
      information about which tokens have changed value and which ones have not.
  \end{itemize}
  We point out that, in both cases, the automaton transitions to the next state in a deterministic fashion.
\item $F_{\SV} = Q_{\SV} \setminus \{ \sinkT \}$ is the set of \emph{final
  states}.
\end{itemize}

Correctness of the DFA $\mathcal T_{\SV}$ is proved by the next
lemma.

\begin{lemma}
  Let $P = (\SV, S)$ be an eager qualitative timeline-based planning problem.
  Then, words accepted by $\mathcal T_{\SV}$ are exactly those encoding plans
  over \SV.
  Moreover the size of $\mathcal T_{\SV}$ is at most exponential in the size of
  $P$.
\end{lemma}



\section{DFA Accepting Solution Plans}
\label{sec:dfaForSolutions}

In this section, we go through the construction of an automaton recognizing solution plans for a planning problem. Towards that, it will come in handy to define some auxiliary structures, namely \emph{blueprints}, \emph{snapshots} and \emph{viewpoints}; moreover, we will define how these structures evolve and give a high-level intuition for each of them.

Let $P = (\SV, S)$ be an eager qualitative timeline-based planning problem,
and let $V = \cup_{x\in \SV}V_x$. We first show how to build a DFA
$\automAP$, whose size is at most exponential in the size of $P$, that
accepts exactly those words encoding solutions plans for $P$ when
restricted to words encoding plans over \SV. In different terms, if a word
encodes a plan over \SV, then it is accepted by \automAP if and only if it
encodes a solution plan for $P$. However, \automAP may also accept words
that do not encode a plan over \SV. Therefore, we need the intersection of
such a DFA \automAP with DFA $\mathcal T_{\SV}$ from the previous section.

In the following, we use \emph{preorders} to represent the ordering relation imposed by synchronization rules.
Each existential statement of the form $\exists a_1[x_1=v_1]a_2[x_2=v_2]
\ldots a_n[x_n=v_n]. \rulebody$, with $\rulebody$ conjunction of atoms,
identifies a \emph{preorder} whose domain is the set of terms
$\tokstart(a)/\tokend(a)$ occurring in $\rulebody$, and where term $t_1$
precedes term $t_2$ in the preorder whenever $t_1 \before t_2$ belongs to
$\hat \rulebody$. 

For a preorder $\preorderP$, we denote by $\domainP$ its domain and by $\preceqP$ the ordering relation. Moreover, we use $x \equivP y$ to denote the fact that both $x \preceqP y$ and $y \preceqP x$ hold, and $x \precP y$ to denote the fact that $x \preceqP y$ holds but $y \preceqP x$ does not. Finally, we denote by $[x]_{\equivP}$ the equivalence class of $x$ with respect to $\equivP$ for every $x \in \domainP$, that is, $[x]_{\equivP} = \{ y \in \domainP \mid y \equivP x \}$. We omit the subscript $_\preorderP$ when it is clear from the context. A preorder $\preorderP$ induces a directed acyclic graph (DAG) $G=(V,A)$, where $V$ is the set of equivalence classes, that is, $V = \{ [x]_\equiv \mid x \in \domainP \}$, and, for every $x,y \in \domainP$ there is an arc from $[x]_\equiv$ to $[y]_\equiv$ in $A$ (denoted by $([x]_\equiv, [y]_\equiv) \in A$ or $[x]_\equiv \rightarrow [y]_\equiv$ when set $A$ is clear from the context) if and only if $x \prec y$ and there is no $w \in \domainP$ such that $x \prec w$ and $w \prec y$. Clearly, there is a path from $[x]_\equiv$ to $[y]_\equiv$ (denoted by $[x]_\equiv \rightarrow^* [y]_\equiv$) if and only if $x \preceq y$.
Therefore, given an existential statement $\E$ occurring in a synchronization rule $\Rule$, we refer to the associated preorder and DAG as, respectively, $\preorderE$ and $G_\E$.

It is important to observe that a conjunction of atoms $\rulebody$ within
an existential statement \E contains atoms of both forms $t_1 \before t_2$
and $t_1 \before* t_2$. To keep track of these different constraints in DAG
$G_\E = (V,A)$ associated with $\E$, we identify the subset $A_< \subseteq
A$ of arcs of $G_\E$ as the set $A_< = \{ ([x]_\equiv , [y]_\equiv) \in A \mid x \before* y \in \hat \rulebody \}$. We sometimes write $[x]_\equiv \Rightarrow [y]_\equiv$ for $([x]_\equiv,[y]_\equiv) \in A_<$, when $A$ is clear from the context. Figure 1 shows such a difference. 


\begin{figure}[t]
\centering
\begin{tikzpicture}

\draw (0,0) circle (0.8cm);

\draw (3,0) circle (0.8cm);

\draw (6,0) circle (0.8cm);

\draw [->,>=stealth, double distance = 0.05cm] (0.8,0) -- (2.2,0) node [midway,above] {};

\draw [->,>=stealth] (3.8,0) -- (5.2,0) node [midway,above] {};

\node [align=center] at (0,0) {$\tokstart(a_0)$ \\ $ \tokstart(a_1)$};
\node at (3,0) {$\tokend(a_0)$};
\node at (6,0) {$\tokend(a_1)$};

\draw (0,-2) circle (0.8cm);

\draw (3,-2) circle (0.8cm);

\draw (6,-2) circle (0.8cm);

\draw [->,>=stealth, double distance = 0.05cm] (0.8,-2) -- (2.2,-2) node [midway,above] {};

\draw [->,>=stealth, double distance = 0.05cm] (3.8,-2) -- (5.2,-2) node [midway,above] {};

\node [align=center] at (0,-2) {$\tokstart(a_0)$ \\ $\tokstart(a_1)$};
\node at (3,-2) {$\tokend(a_0)$};
\node at (6,-2) {$\tokend(a_1)$};
\end{tikzpicture}

\label{fig:dag-begin-strict-nonstrict}
\caption{
  Above, we show the blueprint for the unique existential statement in the rule
  $a_0[x_0=v_0] \implies \exists a_1[x_1=v_1].
  (\tokstart(a_0) = \tokstart(a_1) \wedge \tokend(a_0) \before \tokend(a_1))$,
  from \cref{sec:fragment}.
  It forces token $a_0$ to either be a prefix of or coincide with token $a_1$.
  Below, the blueprint obtained replacing $\tokend(a_0) \before \tokend(a_1)$
  with $\tokend(a_0) \before* \tokend(a_1)$, that forces $a_1$ to be a strict
  prefix of $a_1$.}
\end{figure}

Let \E be an existential statement occurring in a rule \Rule and $G_\E$ the DAG associated with \E. The set of \emph{events associated with a vertex} $[x]_\equiv$ of $G_\E$, denoted by \eventsTwoGExequiv, is the smallest set such that if $\tokstart(a) \in [x]_\equiv$ (resp., $\tokend(a) \in [x]_\equiv$) and $a[y=v]$ either occurs in \E or is the trigger of \Rule, then $\mathit{start}(y,v) \in \eventsTwoGExequiv$ (resp., $\mathit{end}(y,v) \in \eventsTwoGExequiv$). The set of \emph{events associated with a subset $V'$ of vertices} of $G_\E$, denoted by \eventsTwoGEVprime, is the set $\bigcup_{v \in V'} \eventsTwoGEv$.

\subsection{Blueprints, Snapshots, and Viewpoints}
A DAG associated with an existential statement \E is also called a \emph{blueprint} for \E. A \emph{snapshot} for an existential statement \E is a pair $(G, K)$, where $G = (V,A)$ is a blueprint for \E and $K \subseteq V$ is a \emph{downward closed} subset of vertices of $G$, that is, $v \in K$ implies $v' \in K$ for all $v' \in V$ with $v' \rightarrow^* v$. The number of different snapshots for \E is at most $2^{|V|}$, hence at most exponential in the size of $P$, denoted by $|P|$. A \emph{viewpoint} \viewpointV for a rule \Rule is a set of snapshots for existential statements in \Rule, at most one for each statement. Let $n_\Rule$ be the number of existential statements in \Rule; then, it is easy to see that the number of different viewpoints for \Rule is at most $(2^{|P|})^{n_{\Rule}}$, hence exponential in the size of $P$. If $K = \emptyset$ for all $(G,K) \in \viewpointV$, then \viewpointV is the \emph{initial} viewpoint of \Rule; analogously, if $K$ is the entire set of vertices of $G$, for some $(G,K) \in \viewpointV$, then \viewpointV is a \emph{final} viewpoint of \Rule.

Intuitively, a viewpoint checks the satisfaction of a rule \Rule by recognizing when at least one existential statement has been fulfilled. This check works by collecting, for each existential statement, information about the tokens seen so far along the plan into snapshots, which are downward closed and accurately represent all relevant symbols read. How information is collected, thus how viewpoints and snapshots evolve, is explained in the following.

States of automata \automAP are sets of viewpoints containing at least one
viewpoint for each rule of $P$ (besides a fresh rejecting sink state \sinkAP);
recall that viewpoints are in turn sets of snapshots.
Therefore, to define automata runs, we first show how snapshots and
viewpoints evolve upon reading an alphabet symbol.
To this end, we need the following notions.

For a snapshot $(G,K)$, we set $\nextTwoGK = K'$, where $K'$ is the largest
downward closed subset of vertices of $G$ for which there is no pair of vertices
$v,v' \in K'\setminus K$ with $v \Rightarrow v'$.
Moreover, given an alphabet symbol $\sigma \in \Sigma_\SV$, we define
$\nextThreeGKsigma = K'$, where $K'$ is the largest downward closed subset of
vertices of $\nextTwoGK$ such that $\eventsTwoG{K' \setminus K} \subseteq
\eventsOnesigma$.
We say that snapshot $(G,K)$ is \emph{compatible} with symbol $\sigma$ if for
all $\mathit{start}(x,v) \in \eventsTwoGK$ and $\mathit{end}(x,v) \in
\eventsOnesigma \cap \eventsTwoG{V \setminus K}$, it holds that
$\mathit{end}(x,v) \in \eventsTwoG{\nextThreeGKsigma}$.
%
%

Intuitively, during a run of the automaton, a snapshot $(G,K)$ evolves by
suitably extending $K$.
$\nextTwoGK$ identifies the only vertices that can appear in such an
extension independently from the alphabet symbol read, that is, vertices in $V
\setminus K$ reachable (from $K$) without crossing arcs in $A_{<}$.
The exact extension, however, depends on the actual symbol $\sigma$ read by the
automaton: $K$ cannot be extended with events that are not included in $\sigma$.
Therefore, \nextThreeGKsigma identifies precisely how a snapshot evolves.
%
%
At last, observe that for a snapshot to be allowed to evolve upon
reading a symbol, it must be guaranteed that no token ending is overlooked,
which is formalized by the notion of compatibility of a snapshot with a symbol.

We can now characterize the evolution of snapshots and viewpoints when reading an alphabet symbol $\sigma \in \Sigma_\SV$. The \emph{evolution of a snapshot} $(G,K)$ when reading $\sigma$, denoted \evolsnapGKsigma, is snapshot $(G, \nextThreeGKsigma)$, if $(G,K)$ is compatible with $\sigma$; it is undefined otherwise. The \emph{evolution of a viewpoint} \viewpointV when reading $\sigma$, denoted \evolviewVsigma, is viewpoint $\viewpointVprime$, defined as the smallest set such that for all $(G,K) \in \viewpointV$, if \evolsnapGKsigma is defined, then $\evolsnapGKsigma \in \viewpointVprime$.

\subsection{States, Initial State, and Final States of \texorpdfstring{\automAP}{AP}}

We have already mentioned that \emph{states} of \automAP are sets of viewpoints
containing at least one viewpoint for each rule $\Rule \in S$ (recall that $S$
is the set of rules in planning problem $P$), besides a fresh rejecting sink
state \sinkAP.
However, since it is crucial for us to bound the size of \automAP to be at most
exponential in the one of $P$, we impose the \emph{linearity condition},
formalized in what follows.

First, recall that, given a rule \Rule, featuring existential statements $\E_1, \ldots, \E_{n_\Rule}$, a viewpoint \viewpointV for \Rule only contains at most one snapshot for each existential statement in \Rule; therefore, it holds that $|\viewpointV| \leq n_\Rule$ and there is a partial surjective function $f_\viewpointV : \{ \E_1, \ldots, \E_{n_\Rule} \} \rightarrow \viewpointV$, where $f_\viewpointV(\E)$ is the only snapshot for \E in \viewpointV, if any, for all $\E \in \{ \E_1, \ldots, \E_{n_\Rule} \}$.

Now, for all rules $\Rule \in S$, let \viewpointsetR be the set of viewpoints
for \Rule, and let $\viewpointsetP = \bigcup_{\Rule \in S} \viewpointsetR$.
We define an ordering relation $\preceq$ between viewpoints: for all
$\viewpointV, \viewpointVprime \in \viewpointsetP$, it holds that $\viewpointV
\preceq \viewpointVprime$ if and only if
\begin{enumerate*}[label={\it (\roman*)}]
\item $\viewpointV, \viewpointVprime \in \viewpointsetR$ for some $\Rule \in S$,
\item $\mathit{dom}(f_\viewpointVprime) \subseteq
  \mathit{dom}(f_\viewpointV)$,\footnote{For a partial function $f$, we denote by
    $\mathit{dom}(f)$ the set of elements where $f$ is defined.} and
\item for all $\E \in \mathit{dom}(f_{\viewpointVprime})$, we have that
  $f_{\viewpointV}(\E) = (G,K)$, $f_{\viewpointVprime}(\E) = (G,K')$, and $K
  \subseteq K'$.
\end{enumerate*}
%
%
Intuitively, $\viewpointV \preceq \viewpointVprime$ captures the fact that
\viewpointVprime has gone further than \viewpointV in matching input symbols to
satisfy a rule.
Therefore, a snapshot in \viewpointV either evolved into one in
\viewpointVprime, according to the symbols read, or has disappeared because it
is not compatible with some of the symbols read, and thus it cannot be used
anymore to satisfy the rule.


At this point, we can formalize the linearity condition, crucial to constrain
the size of \automAP (Lemma~\ref{lem:automatonAP}).

\begin{definition}[Linearity condition] A set of viewpoints $\viewpointset$ satisfies the \emph{linearity condition} if for all viewpoints $\viewpointV, \viewpointVprime \in \viewpointset$ and rules $\Rule \in S$, if $\viewpointV, \viewpointVprime \in \viewpointsetR$, then $\viewpointV \preceq \viewpointVprime$ or $\viewpointVprime \preceq \viewpointV$ holds. \end{definition}
\noindent Intuitively, we impose all viewpoints for the same rule in a state of \automAP to be linearly ordered.

We are now ready to formally characterize the set of \emph{states} of \automAP,
consisting of the sets $\viewpointset \subseteq \viewpointsetP$ of viewpoints
that contain at least one viewpoint for each rule $\Rule \in S$ and that satisfy
the linearity condition, and including, in addition, a fresh \emph{rejecting
  sink state} \sinkAP.
We denote it by $Q_P$.

The \emph{initial state} $q^0_P$ of \automAP is the set $\{ \viewpointV^0_{\Rule} \mid \Rule \in S \}$, where $\viewpointV^0_\Rule$ is the initial viewpoint of rule \Rule.

Towards a definition of the set $F_P$ of \emph{final states} of \automAP, we introduce the notion of \emph{enabled viewpoints}. A viewpoint \viewpointV for rule $\Rule \in S$ is \emph{enabled} if either \Rule is triggerless or \Rule has trigger token $a_0$ and $\tokstart(a_0) \in K$ for some $(G,K) \in \viewpointV$. A state $q$ of \automAP is \emph{final} if every enabled viewpoint therein is final.

\subsection{Transition Function of \texorpdfstring{\automAP}{AP}}

The last step of our construction is the definition of the \emph{transition function} $\delta_P$ for automaton \automAP.

To this end, we introduce the notion of alphabet symbol \emph{enabling} a
viewpoint \viewpointV along with the one of states of \automAP \emph{compatible}
with an alphabet symbol.
Let \viewpointV be a viewpoint for a non-triggerless rule \Rule with trigger
token $a_0$ and $\sigma \in \Sigma_\SV$ an alphabet symbol.
We say that $\sigma$ \emph{enables} \viewpointV if there is $(G,K) \in
\viewpointV$ with $\tokstart(a_0) \in \nextThreeGKsigma$.
Moreover, we say that a state $q \in Q_P \setminus \{ \sinkAP \}$ is
\emph{compatible} with $\sigma$ if for all non-triggerless rules $\Rule \in S$,
with trigger token $a_0[x_0 = v_0]$, if $\mathit{start}(x_0,v_0) \in
\eventsOnesigma$, then there is a viewpoint $\viewpointV \in q$ such that
$\sigma$ enables \viewpointV.

We are now ready to define the \emph{transition function} $\delta_P$ of \automAP. For all $q \in Q_P$ and alphabet symbol $\sigma \in \Sigma_\SV$: \begin{itemize}
\item if $q = \sinkAP$ or $q$ is not compatible with $\sigma$, then $\delta(q,\sigma) = \sinkAP$; \item otherwise, $\delta(q,\sigma) = q'$, where $q'$ is the smallest set such that for all $\viewpointV \in q$
    \begin{itemize} 
    \item $\evolviewVsigma \in q'$ and 
    \item if $\sigma$ enables \viewpointV, then $\viewpointV \in q'$.
    \end{itemize}
\end{itemize}

\begin{lemma}
  \label{lem:automatonAP}
  Let $P = (\SV, S)$ be an eager qualitative timeline-based planning problem.
  Then, each finite word over $\Sigma_\SV$ that encodes a plan over \SV is
  accepted by \automAP if and only if it encodes a solution plan for $P$.
  Moreover, the size of $\automAP$ is at most exponential in the size of $P$.
\end{lemma}
\begin{proof}
  For lack of space, we omit the proof of soundness showing that the automaton
  accepts the correct language as claimed.
  Instead, we show that the size of $\automAP$ is indeed at most exponential in
  the size of $P$.

  Let $k$ be the largest number of existential statements in a rule of $P$ and $k'$ the largest number of atoms in an existential statement of $P$. Thanks to the linearity rule enjoyed by states of $P$, it is not difficult to convince oneself that the number of different viewpoints for the same rule in a state $q \in Q_P$ to be at most $k \times k'$. Thus, each state in $Q_P$ contains at most $|S| \times k \times k'$ different viewpoints (the product of the number of rules in $P$ by the number of different viewpoints for the same rule).

  Therefore, the size of $Q_P$ is at most $|\viewpointsetP|^{(|S| \times k
    \times k')}$.
  Clearly, $(|S| \times k \times k')$ is at most polynomial in the size of $P$.
  Since $|\viewpointsetP| \leq \sum_{\Rule \in S} |\viewpointsetR|$ and, as
  already pointed out, $|\viewpointsetR|$ is at most exponential in the size of
  $P$, we can conclude that the size of $Q_P$ is at most exponential in the size
  of $P$.
\end{proof}

\begin{theorem}\label{thm:automaton}
  Let $P = (\SV, S)$ be an eager qualitative timeline-based planning problem.
  Then, the words accepted by the intersection automaton of $\automAP$ and
  $\mathcal T_{\SV}$ are exactly those encoding solution plans for $P$.
  Moreover, the size of the intersection automaton of $\automAP$ and $\mathcal
  T_{\SV}$ is at most \emph{exponential} in the size of $P$.
\end{theorem}






\section{A Maximal Subset of Allen's Relations} \label{sec:allen}


Allen's interval algebra is a formalism for temporal reasoning introduced in~\cite{Allen83}. 
It identifies all possible relations between pairs of time intervals over a linear order and specifies a machinery to reason about them.
%
In this section, we isolate the maximal subset of Allen's relations captured by
the eager fragment of qualitative timeline-based planning. To this end, we show
how to map Allen's relations over tokens in terms of their endpoints, that is,
as conjunctions of atoms over terms $\tokstart(a), \tokstart(b), \tokend(a),
\tokend(b)$, for token names $a$ and $b$.
%
%
Then, we check which relation encoding satisfies the conditions of
\cref{def:eager:rule}.
%
%
Let $a, b \in \toknames$.
\begin{itemize}
\item$a \ibefore b$ ($b \iafter a$) can be defined as $\tokend(a) \before* \tokstart(b)$.

\item$a \meets b$ ($b \ismet a$) can be defined as $\tokend(a) = \tokstart(b)$.

\item$a \ends b$ ($b \isended a$) can be defined as $\tokstart(b) \before* \tokstart(a) \land \tokend(a) = \tokend(b)$.

\item$a \starts b$ ($b \isstarted a$) can be defined as $\tokstart(a) = \tokstart(b) \land \tokend(a) \before* \tokend(b)$.

\item$a \overlaps b$ ($b \isoverlapped a$) can be defined as $\tokstart(a) \before* \tokstart(b) \land \tokstart(b) \before* \tokend(a) \land \tokend(a) \before* \tokend(b)$.

\item$a \during b$ ($b \icontains a$) can be defined as $\tokstart(b) \before* \tokstart(a) \land \tokend(a) \before* \tokend(b)$.

\item$a \isequal b$ can be defined as $\tokstart(a) = \tokstart(b) \land \tokend(a) = \tokend(b)$.
\end{itemize}

It is not difficult to see that, if one of the tokens, let's say $a$, is the
trigger token, then the encodings not complying with \cref{def:eager:rule} are
the ones for Allen's relations $\ends$, $\isended$, $\overlaps$,
$\isoverlapped$, and $\during$.
Thus, the maximal subset of Allen's relations that can be captured by an
instance of the eager fragment of the timeline-based planning problem consists
of relations $\ibefore$, $\iafter$, $\meets$, $\ismet$, $\starts$, $\isstarted$,
$\icontains$, and $\isequal$.

As an example, consider relation $\overlaps$ and let $\rulebody = \{
\tokstart(a) \before* \tokstart(b), \tokstart(b) \before* \tokend(a),
\tokend(a) \before* \tokend(b) \}$ be its encoding.
Clearly, the transitive closure $\hat{\rulebody}$ of \rulebody
(cf. Section~\ref{sec:fragment}) includes also $\tokstart(a) \before
\tokstart(b)$ and $\tokend(a) \before \tokend(b)$ but it does not include
$\tokstart(b) \before \tokstart(a)$, thus violating Condition 2 of
\cref{def:eager:rule}.
A similar argument can be used for relations $\ends$, $\isended$,
$\isoverlapped$, and $\during$.

If, instead, none of the token is a trigger token, then the only Allen's
relations not violating any of the conditions of \cref{def:eager:rule} are
$\ibefore$, $\iafter$, $\meets$, and $\ismet$.
We omit the details.







\section{Conclusions}
\label{sec:conclusion}

In this paper, 
we identified a meaningful fragment of timeline-based planning
whose solutions
can be recognized by DFAs of singly exponential size. Specifically, we
identified restrictions on the allowed synchronization rules, which we
called \emph{eager rules}, for which we showed how to build the
corresponding deterministic automaton of exponential size, that
can then be directly exploited to synthesize strategies.
Moreover, we isolated a maximal subset of Allen's relations captured by
such a fragment. 

Whether the fragment of timeline-based planning identified by the eager rules is
maximal or not is an open question currently under study.
Further research directions include
%
a parametrized complexity analysis over the number of synchronization rules and
a characterization in terms of temporal logics, like the one
in~\cite{della2017bounded}.


\bibliographystyle{eptcs}
\bibliography{biblio}
\end{document}